\RequirePackage{fix-cm}        
\documentclass[a4paper,runningheads]{llncs}

\usepackage[utf8]{inputenc}
\usepackage{amsmath}
\usepackage{amssymb}
\usepackage{calc}
\usepackage{geometry}
\usepackage{graphicx}
\usepackage{grffile}
\usepackage{hyperref}
\usepackage{bm}
\usepackage[shortlabels,inline]{enumitem}
\usepackage{ifthen}
\usepackage{microtype}
\usepackage{subfig}
\usepackage{todonotes}
\usepackage{xargs}
\usepackage{xfrac}
\usepackage{xspace}

\usepackage{listings}
\newcommandx*\disguisemath[3][1=c]{%
	\mathchoice{%
		\makebox[\widthof{$\displaystyle#2$}][#1]{$\displaystyle#3$}%
	}{%
		\makebox[\widthof{$\textstyle#2$}][#1]{$\textstyle#3$}%
	}{%
		\makebox[\widthof{$\scriptstyle#2$}][#1]{$\scriptstyle#3$}%
	}{%
		\makebox[\widthof{$\scriptscriptstyle#2$}][#1]{$\scriptscriptstyle#3$}%
	}%
}

\usepackage[bibstyle=numeric-comp,citestyle=numeric-comp,firstinits=true,doi=true,isbn=true,maxbibnames=99]{biblatex}
\bibliography{references}

\newtheorem{fact}{Fact}

\newcommand*{\R}{\mathbb{R}}

\newcommandx*{\LDAUOmicron}[2][1=@pkling_false]{\mathrm{O}\!\ifthenelse{\equal{#1}{small}}{\bigl(#2\bigr)}{\left(#2\right)}}
\newcommandx*{\LDAUomicron}[2][1=@pkling_false]{\mathrm{o}\ifthenelse{\equal{#1}{small}}{\bigl(#2\bigr)}{\left(#2\right)}}
\newcommandx*{\LDAUOmega}[2][1=@pkling_false]{\Omega\!\ifthenelse{\equal{#1}{small}}{\bigl(#2\bigr)}{\left(#2\right)}}
\newcommandx*{\LDAUomega}[2][1=@pkling_false]{\omega\ifthenelse{\equal{#1}{small}}{\bigl(#2\bigr)}{\left(#2\right)}}
\newcommandx*{\LDAUTheta}[2][1=@pkling_false]{\Theta\ifthenelse{\equal{#1}{small}}{\bigl(#2\bigr)}{\left(#2\right)}}

\DeclareMathOperator{\dif}{d\!}

\newcommand*{\od}[3][]{\ensuremath{\ifinner\tfrac{\dif{^{#1}}#2}{\dif{#3^{#1}}}\else\dfrac{\dif{^{#1}}#2}{\dif{#3^{#1}}}\fi}}

\newcommand*{\abs}[1]{\left\lvert#1\right\rvert}

\newcommand*{\intco}[1]{\ensuremath{\left[#1\right)}}
\newcommand*{\intoc}[1]{\ensuremath{\left(#1\right]}}
\newcommand*{\intcc}[1]{\ensuremath{\left[#1\right]}}
\newcommandx*{\set}[2][2=@pkling_false]{\left\{#1\ifthenelse{\equal{#2}{@pkling_false}}{}{\;\middle|\;#2}\right\}}

\newcommand*{\AVR}{\ensuremath{\text{AVR}}\xspace}
\newcommand*{\BKP}{\ensuremath{\text{BKP}}\xspace}
\newcommand*{\cost}[1]{\operatorname{cost}(#1)}
\newcommand*{\crejONE}{\ensuremath{c_{1}}\xspace}
\newcommand*{\crejTWO}{\ensuremath{c_{2}}\xspace}
\newcommand*{\valdens}[1]{\ensuremath{\delta_{#1}}}
\newcommand*{\valdensMAX}{\ensuremath{\delta_{\text{max}}}\xspace}
\newcommand*{\Eidle}[1]{\ensuremath{E^{#1}_{\text{idle}}}}
\newcommand*{\Esys}[1]{\ensuremath{E^{#1}_{\text{sys}}}}
\newcommand*{\Esleep}[1]{\ensuremath{E^{#1}_{\text{sleep}}}}
\newcommand*{\Ework}[1]{\ensuremath{E^{#1}_{\text{work}}}}
\newcommand*{\OA}{\ensuremath{\text{OA}}\xspace}
\newcommand*{\penrat}[1]{\ensuremath{\Gamma_{#1}}}
\newcommand*{\penratMAX}{\ensuremath{\Gamma}\xspace}
\newcommand*{\PS}{\ensuremath{\text{PS}}\xspace}

\newcommand*{\scrit}{\ensuremath{s_{\text{cr}}}\xspace}
\newcommand*{\sprofup}[1]{\ensuremath{s_{#1,\text{p}}}}
\newcommand*{\Vrej}[1]{\ensuremath{V^{#1}_{\text{rej}}}}
\newcommand*{\Wrem}[2]{\ensuremath{w^{#1}_{#2}}}

\title{Slow Down \& Sleep for Profit in\\Online Deadline Scheduling\thanks{This work was partially supported by the German Research Foundation (DFG) within the Collaborative Research Center ``On-The-Fly Computing'' (SFB 901) and by the Graduate School on Applied Network Science (GSANS).}}
\titlerunning{Slow Down \& Sleep for Profit in Online Deadline Scheduling}
\date{\today}
\author{Andreas Cord-Landwehr \and Peter Kling \and Frederik Mallmann-Trenn}
\institute{Heinz Nixdorf Institute and Computer Science Department, University of Paderborn\\
\email{\{andreas.cord-landwehr@,\,peter.kling@,\,xarph@mail.\}uni-paderborn.de}
}

\begin{document}
\maketitle
\begin{abstract}
We present and study a new model for energy-aware and profit-oriented scheduling on a single processor.
The processor features dynamic speed scaling as well as suspension to a sleep mode.
Jobs arrive over time, are preemptable, and have different sizes, values, and deadlines.
On the arrival of a new job, the scheduler may either accept or reject the job.
Accepted jobs need a certain energy investment to be finished in time, while rejected jobs cause costs equal to their values.
Here, power consumption at speed $s$ is given by $P(s)=s^{\alpha}+\beta$ and the energy investment is power integrated over time.
Additionally, the scheduler may decide to suspend the processor to a sleep mode in which no energy is consumed, though awaking entails fixed transition costs $\gamma$.
The objective is to minimize the total value of rejected jobs plus the total energy.

Our model combines aspects from advanced energy conservation techniques (namely speed scaling and sleep states) and profit-oriented scheduling models.
We show that \emph{rejection-oblivious} schedulers (whose rejection decisions are not based on former decisions) have – in contrast to the model without sleep states – an unbounded competitive ratio w.r.t\text{.} the processor parameters $\alpha$ and $\beta$.
It turns out that the worst-case performance of such schedulers depends linearly on the jobs' value densities (the ratio between a job's value and its work).
We give an algorithm whose competitiveness nearly matches this lower bound.
If the maximum value density is not too large, the competitiveness becomes $\alpha^{\alpha}+2e\alpha$.
Also, we show that it suffices to restrict the value density of low-value jobs only.
Using a technique from \cite{Chan:2010} we transfer our results to processors with a fixed maximum speed.
\end{abstract}

\section{Introduction}
Over the last decade, energy usage of data centers and computers in general has become a major concern.
There are various reasons for this development: the ubiquity of technical systems, the rise of mobile computing, as well as a growing ecological awareness.
Also from an economical viewpoint, energy usage can no longer be ignored.
Energy costs for both the actual computation and the cooling have become \emph{the} decisive cost factor in today's data centers (see, e.g., \textcite{Barroso:2007}).
In combination with improvements on the technical level, algorithmic research has great potential to reduce energy consumption.
\Textcite{Albers:2010} gives a good insight on the role of algorithms to fully exploit the energy-saving mechanisms of modern systems.
Two of the most prominent techniques for power saving are \emph{dynamic speed scaling} and \emph{power-down}.
The former allows a system to save energy by adapting the processor's speed to the current system load, while the latter can be used to transition into a sleep mode to conserve energy.
There is an extensive body of literature on both techniques (see below).
From an algorithmic viewpoint, the most challenging aspect in the design of scheduling strategies is to handle the lack of knowledge about the future: should we use a high speed to free resources in anticipation of new jobs or enter sleep mode in the hope that no new jobs arrive in the near future?

Given that profitability is a driving force for most modern systems and that energy consumption has gained such a high significance, it seems natural to take this relation explicitly into account.
\Textcite{Pruhs:2010} consider a scheduling model that does so by introducing job values.
Their scheduler controls energy usage via speed scaling and is allowed to reject jobs if their values seem too low compared to their foreseeable energy requirements.
The objective is to maximize the profit, which is modeled as the total value of finished jobs minus the invested energy.
Our work is based on a result by \textcite{Chan:2010}.
We enhance their model by combining speed scaling and power-down mechanisms for energy management, which not only introduces non-trivial difficulties to overcome in the analysis, but proves to be inherently more complex compared to the original model insofar that classical algorithms can become arbitrarily bad.

\subsubsection{History \& Related Work.}
There is much literature concerning energy-aware scheduling strategies both in practical and theoretical contexts.
A recent survey by \textcite{Albers:2011} gives a good and compact overview on the state of the art in the dynamic speed scaling setting, also in combination with power-down mechanisms.
In the following, we focus on theoretical results concerning scheduling on a single processor for jobs with deadlines.
Theoretical work in this area has been initiated by \textcite{Yao:1995}.
They considered scheduling of jobs having different sizes and deadlines on a single variable-speed processor.
When running at speed $s$, its power consumption is $P(s)=s^{\alpha}$ for some constant $\alpha\geq2$.
\Citeauthor{Yao:1995} derived a polynomial time optimal offline algorithm as well as two online algorithms known as \emph{optimal available} (\OA) and \emph{average rate} (\AVR).
Up to now, \OA remains one of the most important algorithms in this area, as it is used as a basic building block by many strategies (including the strategy we present in this paper).
Using an elegant amortized potential function argument, \textcite{Bansal:2007a} were able to show that \OA's competitive factor is exactly $\alpha^{\alpha}$.
Moreover, the authors stated a new algorithm, named \BKP, which achieves a competitive ratio of essentially $2e^{\alpha+1}$.
This improves upon \OA for large $\alpha$.
The best known lower bound for deterministic algorithms is $\sfrac{e^{\alpha-1}}{\alpha}$ due to \textcite{Bansal:2009}.
They also presented an algorithm (qOA) that is particularly well-suited for low powers of $\alpha$.
An interesting and realistic model extension is the restriction of the maximum processor speed.
In such a setting, a scheduler may not always be able to finish all jobs by their deadlines.
\Textcite{Chan:2007} were the first to consider the combination of classical speed scaling with such a maximum speed.
They gave an algorithm that is $\alpha^{\alpha}+\alpha^2 4^{\alpha}$-competitive on energy and $14$-competitive on throughput.
\Textcite{Bansal:2008} improved this to a $4$-competitive algorithm concerning the throughput while maintaining a constant competitive ratio with respect to the energy.
Note that no algorithm – even if ignoring the energy consumption – can be better than $4$-competitive for throughput (see~\cite{Baruah:1991}).

Power-down mechanisms were studied by \textcite{Baptiste:2006}.
He considered a fixed-speed processor needing a certain amount of energy to stay awake, but which may switch into a sleep state to save energy.
Returning from sleep needs energy $\gamma$.
For jobs of unit size, he gave a polynomial time optimal offline algorithm, which was later extended to jobs of arbitrary size~\cite{Baptiste:2007}.
The first work to combine both dynamic speed scaling and sleep states in the classical YAO-model is due to \textcite{Irani:2007}.
They achieved a $2$-approximation for arbitrary convex power functions.
For the online setting and power function $P(s)=s^{\alpha}+\beta$ a competitive factor of $4^{\alpha-1}\alpha^{\alpha}+2^{\alpha-1}+2$ was reached.
\Textcite{Han:2010} improved upon this in two respects: they lowered the competitive factor to $\alpha^{\alpha}+2$ and transferred the result to scenarios limiting the maximum speed.
Only recently, \textcite{Albers:2012} proved that the optimization problem is NP-hard and gave lower bounds for several algorithm classes.
Moreover, they improved the approximation factor for general convex power functions to $\sfrac{4}{3}$.
The papers most closely related to ours are due to \textcite{Pruhs:2010} and \textcite{Chan:2010}.
Both considered the dynamic speed scaling model of \citeauthor{Yao:1995}.
However, they extended the idea of energy-minimal schedules to a profit-oriented objective.
In the simplest case, jobs have values (or priorities) and the scheduler is no longer required to finish all jobs.
Instead, it can decide to reject jobs whose values do not justify the foreseeable energy investment necessary to complete them.
The objective is to maximize profit~\cite{Pruhs:2010} or, similarly, minimize the loss~\cite{Chan:2010}.
As argued by the authors, the latter model has the benefit of being a direct generalization of the classical model of \textcite{Yao:1995}.
For maximizing the profit, \textcite{Pruhs:2010} showed that, in order to achieve a bounded competitive factor, resource augmentation is necessary and gave a scalable online algorithm.
For minimizing the loss, \textcite{Chan:2010} gave a $\alpha^{\alpha}+2e\alpha$-competitive algorithm and transferred the result to the case of a bounded maximum speed.

\subsubsection{Our Contribution.}
We present the first model that not only takes into account two of the most prominent energy conservation techniques (namely, speed scaling and power-down) but couples the energy minimization objective with the idea of profitability.
It combines aspects from both \cite{Irani:2007} and \cite{Chan:2010}.
From~\cite{Irani:2007} we inherit one of the most realistic processor models considered in this area:
A single variable-speed processor with power function $P(s)=s^{\alpha}+\beta$ and a sleep state.
Thus, even at speed zero the system is charged a certain amount $\beta$ of energy, but it can suspend to sleep such that no energy is consumed.
Waking up causes transition cost of $\gamma$.
The job model stems from~\cite{Chan:2010}:
Jobs arrive in an online fashion, are preemptable, and have a deadline, size, and value.
The scheduler can reject jobs (e.g., if their values do not justify the presumed energy investment).
Its objective is to minimize the total energy investment plus the total value of rejected jobs.

A major insight of ours is that the maximum value density \valdensMAX (i.e., the ratio between a job's value and its work) is a parameter that is inherently connected to the necessary and sufficient competitive ratio achievable for our online scheduling problem.
We present an online algorithm that combines ideas from \cite{Chan:2010} and \cite{Han:2010} and analyze its competitive ratio with respect to \valdensMAX.
This yields an upper bound of $\alpha^{\alpha}+2 e\alpha+\valdensMAX\frac{\scrit}{P(\scrit)}$.\footnote{The expression $\frac{\scrit}{P(\scrit)}$ depends only on $\alpha$ and $\beta$, see Section~\ref{sec:preliminaries}.}
If the value density of low-valued jobs is not too large or job values are at least $\gamma$, the competitive ratio becomes $\alpha^{\alpha}+2e\alpha$.
Moreover, we show that one cannot do much better: any \emph{rejection-oblivious} strategy has a competitive ratio of at least $\valdensMAX\frac{\scrit}{P(\scrit)}$.
Here, rejection-oblivious means that rejection decisions are based on the \emph{current} system state and job properties only.
This lower bound is in stark contrast to the setting without sleep states, where a rejection-oblivious $\LDAUOmicron{1}$-competitive algorithm exists~\cite{Chan:2010}.
Using the definition of a job's penalty ratio (due to \textcite{Chan:2010}), we extend our results to processors with a bounded maximum speed.

\section{Model \& Preliminaries}\label{sec:preliminaries}
We are given a speed-scalable processor that can be set to any speed $s\in[0,\infty)$.
When running at speed $s$ its power consumption is $P_{\alpha,\beta}(s)=s^{\alpha}+\beta$ with $\alpha\geq2$ and $\beta\geq0$.
If $s(t)$ denotes the processor speed at time $t$, the total power consumption is $\int_{0}^{\infty}P_{\alpha,\beta}(s(t))\dif{t}$.
We can suspend the processor into a sleep state to save energy.
In this state, it cannot process any jobs and has a power consumption of zero.
Though entering the sleep state is for free, waking up needs a fixed \emph{transition energy} $\gamma\geq0$.
Over time, $n$ jobs $J=\set{1,2,\ldots,n}$ are released.
Each job $j$ appears at its release time $r_j$ and has a deadline $d_j$, a (non-negative) value $v_j$, and requires a certain amount $w_j$ of work.
The processor can process at most one job at a time.
Preemption is allowed, i.e., jobs may be paused at any time and continued later on.
If $I$ denotes the period of time (not necessarily an interval) when $j$ is scheduled, the amount of work processed is $\int_Is(t)\dif{t}$.
A job is finished if $\int_Is(t)\dif{t}\geq w_j$.
Jobs not finished until their deadline cause a cost equal to their value.
We call such jobs \emph{rejected}.
A schedule $S$ specifies for any time $t$ the processor's state (asleep or awake), the currently processed job (if the processor is awake), and sets the speed $s(t)$.
W.l.o.g. we assume $s(t)=0$ when no job is being processed.
Initially, the processor is assumed to be asleep.
Whenever it is neither sleeping nor working we say it is \emph{idle}.
A schedule's cost is the invested energy (for awaking from sleep, idling, and working on jobs) plus the loss due to rejected jobs.
Let $m$ denote the number of sleep intervals, $l$ the total length of idle intervals, and $\mathcal{I}_{\text{work}}$ the collection of all working intervals (i.e., times when $s(t)>0$).
Then, the schedule's \emph{sleeping energy} is $\Esleep{S}:=(m-1)\gamma$, its \emph{idling energy} is $\Eidle{S}:=l\beta$, and its \emph{working energy} is $\Ework{S}:=\int_{\mathcal{I}_{\text{work}}}P_{\alpha,\beta}(s(t))\dif{t}$.
We use $\Vrej{S}$ to denote the total value of rejected jobs.
Now, the cost of schedule $S$ is
\begin{equation}
\cost{S}:=\Esleep{S}+\Eidle{S}+\Ework{S}+\Vrej{S}
.
\end{equation}
We seek online strategies yielding a provably good schedule.
More formally, we measure the quality of online strategies by their competitive factor:
For an online algorithm $A$ and a problem instance $I$ let $A(I)$ denote the resulting schedule and $O(I)$ an optimal schedule for $I$.
Then, $A$ is said to be $c$-competitive if $\sup_I\frac{\cost{A(I)}}{\cost{O(I)}}\leq c$.

We define the \emph{system energy} \Esys{S} of a schedule to be the energy needed to hold the system awake (whilst idling and working).
That is, if $S$ is awake for a total of $x$ time units, $\Esys{S}=x\beta$.
Note that $\Esys{S}\leq\Eidle{S}+\Ework{S}$.
The \emph{critical speed} of the power function is defined as $\scrit:=\arg\min_{s\geq0}\sfrac{P_{\alpha,\beta}(s)}{s}$ (cf.~also~\cite{Irani:2007,Han:2010}).
If job $j$ is processed at constant speed $s$ its energy usage is $w_j\cdot\sfrac{P_{\alpha,\beta}(s)}{s}$.
Thus, assuming that $j$ is the only job in the system and ignoring its deadline, \scrit is the energy-optimal speed to process $j$.
One can easily check that $\scrit^{\alpha}=\frac{\beta}{\alpha-1}$.
Given a job $j$, let $\valdens{j}:=\sfrac{v_j}{w_j}$ denote the job's \emph{value density}.
Following~\cite{Chan:2010} and~\cite{Pruhs:2010}, we define the \emph{profitable speed} \sprofup{j} of job $j$ to be the maximum speed for which its processing may be profitable.
More formally, $\sprofup{j}:=\max\set{s\geq0}[w_j\cdot\sfrac{P_{\alpha,0}(s)}{s}\leq v_j]$.
Note that the definition is with respect to $P_{\alpha,0}$, i.e., it ignores the system energy.
The profitable speed can be more explicitly characterized by $\sprofup{j}^{\alpha-1}=\valdens{j}$.
It is easy to see that a schedule that processes $j$ at average speed faster than \sprofup{j} cannot be optimal: rejecting $j$ and idling during the former execution phase would be more profitable.
See Figure~\ref{fig:basicnotions} for an illustration of these notions.
\begin{figure}
\subfloat[Our algorithm tries to use job speeds that essentially stay in the shaded interval.]{\includegraphics[width=0.4\linewidth]{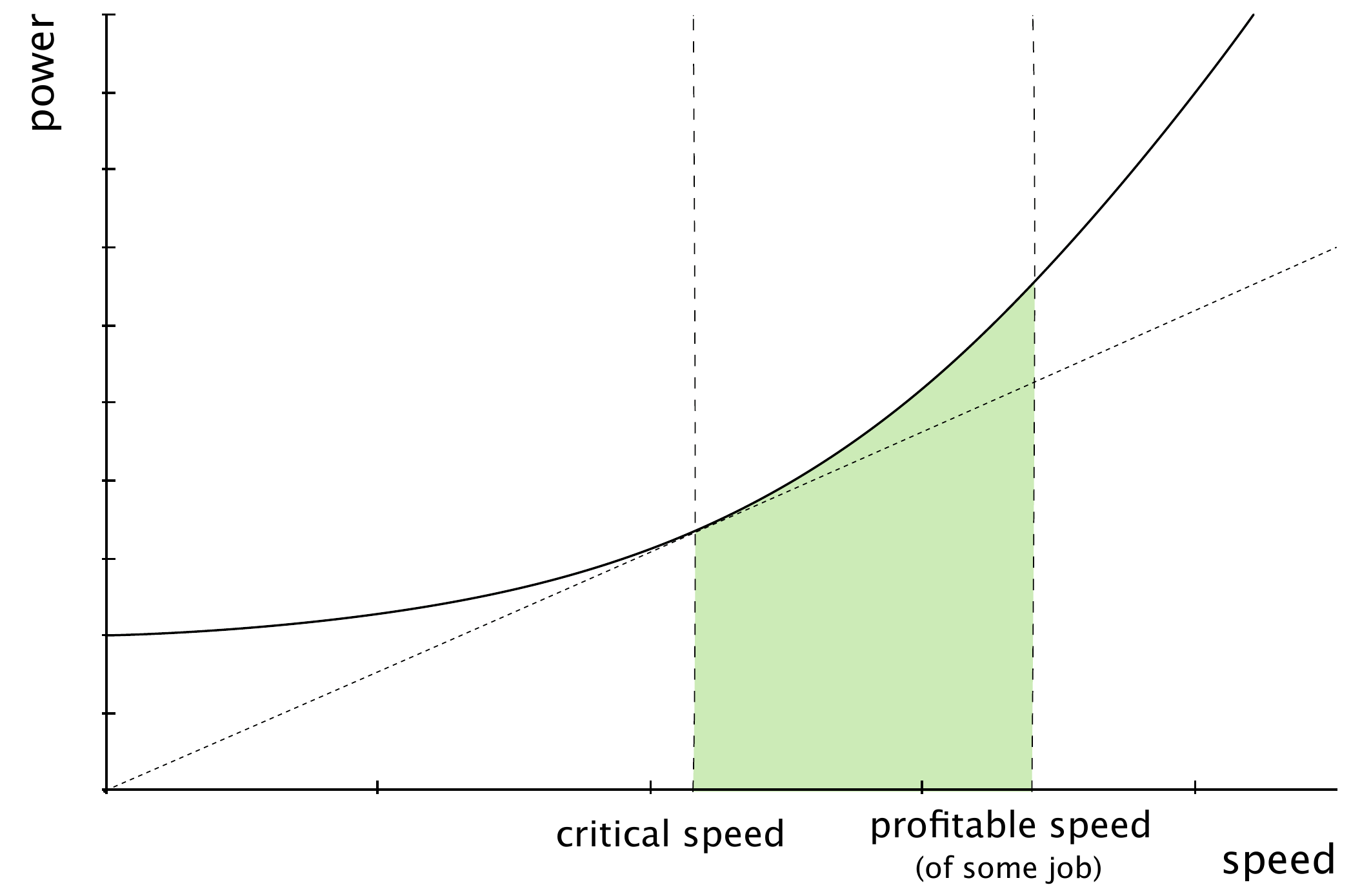}\qquad\quad}
\hfill
\subfloat[A sample schedule and the involved energy types.]{\includegraphics[width=0.49\linewidth]{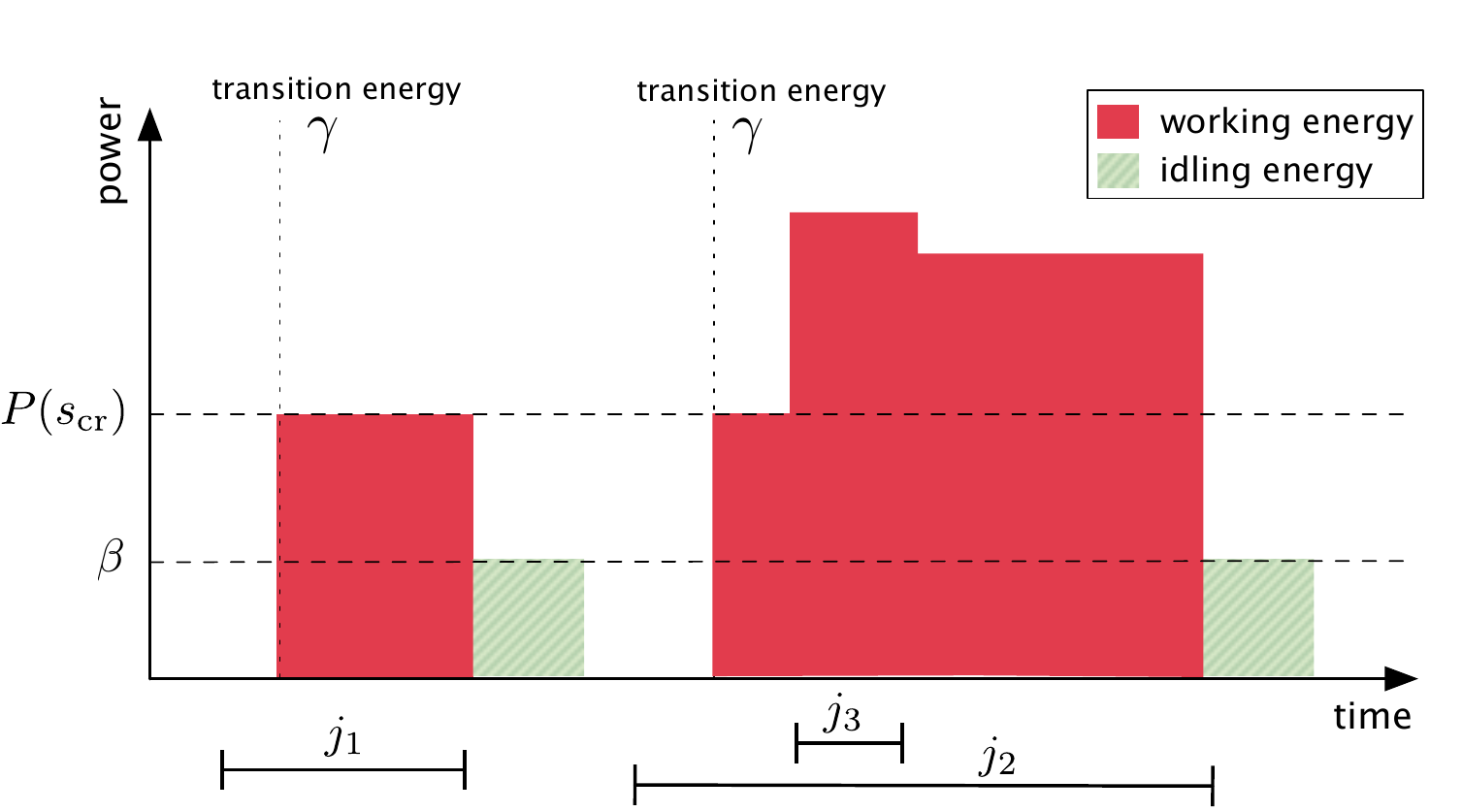}}
\caption{}\label{fig:basicnotions}
\end{figure}

\subsubsection{Optimal Available \& Structural Properties.}
One of the first online algorithms for dynamic speed scaling was Optimal Available (\OA) due to \cite{Yao:1995}.
As it is an essential building block not only of our but many algorithms for speed scaling, we give a short recap on its idea (see~\cite{Bansal:2007a} for a thorough discussion and analysis).
At any time, \OA computes the optimal offline schedule assuming that no further jobs arrive.
This optimal offline schedule is computed as follows:
Let the density of an interval $I$ be defined as $\sfrac{w(I)}{\abs{I}}$.
Here, $w(I)$ denotes the total work of jobs $j$ with $\intco{r_j,d_j}\subseteq I$ and $\abs{I}$ the length of $I$.
Now, whenever a job arrives \OA computes so-called \emph{critical intervals} by iteratively choosing an interval of maximum density.
Jobs are then scheduled at a speed equal to the density of the corresponding critical interval using the earliest deadline first policy.
Let us summarize several structural facts known about the \OA schedule.

\begin{fact}\label{fact:structuralproperties}
Let $S$ and $S'$ denote the \OA schedules just before and after $j$'s arrival.
We use $S(j)$ and $S'(j)$ to denote $j$'s speed in the corresponding schedule.
\begin{enumerate}[(a)]
\item\label{fact:structuralproperties_a} The speed function of $S$ (and $S'$) is a non-increasing staircase function.
\item\label{fact:structuralproperties_b} The minimal speed of $S'$ during \intco{r_j,d_j} is at least $S'(j)$.
\item\label{fact:structuralproperties_c} Let $I$ be an arbitrary period of time during \intco{r_j,d_j} (not necessarily an interval).
	Moreover, let $W$ denote the total amount of work scheduled by $S$ and $W'$ the one scheduled by $S'$ during $I$.
	Then the inequality $W\leq W'\leq W+w_j$ holds.
\item\label{fact:structuralproperties_d} The speed of any $j'\neq j$ can only increase due to $j$'s arrival: $S'(j')\geq S(j')$.
\end{enumerate}
\end{fact}

\section{Lower Bound for Rejection-Oblivious Algorithms}\label{sec:lowerbounds}
This section considers a class of simple, deterministic online algorithms that we call \emph{rejection-oblivious}.
When a job arrives, a rejection-oblivious algorithm decides whether to accept or reject the job.
This decision is based solely on the processor's current state (sleeping, idling, working), its current workload, and the job's properties.
Especially it does not take former decisions into account.
An example for such an algorithm is $PS(c)$ in \cite{Chan:2010}.
For a suitable parameter $c$, it is $\alpha^{\alpha}+2e\alpha$-competitive in a model without sleep state.
In this section we show that in our model (i.e., with a sleep state) no rejection-oblivious algorithm can be competitive.
More exactly, the competitiveness of any such algorithm can become arbitrarily large.
We identify the jobs' value density as a crucial parameter for the competitiveness of these algorithms.

\begin{theorem}\label{thm:lowerbound}
The competitiveness of any rejection-oblivious algorithm $A$ is unbounded.
More exactly, for any $A$ there is a problem instance $I$ with competitive factor $\geq\valdensMAX\frac{\scrit}{P_{\alpha,\beta}(\scrit)}$.
Here, \valdensMAX is the maximum value density of jobs from $I$.
\end{theorem}
\begin{proof}
For $A$ to be competitive, there must be some $x\in\R$ such that, while $A$ is asleep, all jobs of value at most $x$ are rejected (independent of their work and deadlines).
Otherwise, we can define a sequence of $n$ identical jobs $1,2,\ldots,n$ of arbitrary small value $\epsilon$.
W.l.o.g., we release them such that $A$ goes to sleep during \intco{d_{j-1},r_j} (otherwise $A$ consumes an infinite amount of energy).
Thus, $A$'s cost is at least $n\gamma$.
If instead considering schedule $S$ that rejects all jobs, we have $\cost{S}=n\epsilon$.
For $\epsilon\to0$ we see that $A$'s competitive ratio is unbounded.

So, let $x\in\R$ be such that $A$ rejects any job of value at most $x$ whilst asleep.
Consider $n$ jobs of identical value $x$ and work $w$.
For each job, the deadline is set such that $w=\scrit(d_j-r_j)$.
The jobs are released in immediate succession, i.e., $r_j=d_{j-1}$.
Algorithm $A$ rejects all jobs, incurring cost $n x$.
Let $S$ denote the schedule that accepts all jobs and processes them at speed \scrit.
The cost of $S$ is given by $\cost{S}=\gamma+n w\frac{P_{\alpha,\beta}(\scrit)}{\scrit}$.
Thus, $A$'s competitive ratio is at least
\begin{equation*}
\frac{n x}{\gamma+n w\frac{P_{\alpha,\beta}(\scrit)}{\scrit}}=\valdensMAX\frac{1}{\frac{\gamma}{n w}+\frac{P_{\alpha,\beta}(\scrit)}{\scrit}}
.
\end{equation*}
For $n\to\infty$ we get the lower bound $\valdensMAX\smash{\frac{\scrit}{P_{\alpha,\beta}(\scrit)}}$.
\qed
\end{proof}

\section{Algorithm \& Analysis}\label{sec:algorithm+analysis}
In the following, we use $A$ to refer to both our algorithm and the schedule it produces; which is meant should be clear from the context.
As most algorithms in this area (see, e.g.,~\cite{Irani:2007,Bansal:2009,Chan:2010,Han:2010,Albers:2011a}), $A$ relies heavily on the good structural properties of \OA and its wide applicability to variants of the original energy-oriented scheduling model of \textcite{Yao:1995}.
It essentially consists of two components, the \emph{rejection policy} and the \emph{scheduling policy}.
The rejection policy decides which jobs to accept or reject, while the scheduling policy ensures that all accepted jobs are finished until their deadline.
Our rejection policy is an extension of the one used by the algorithm \PS in~\cite{Chan:2010}.
It ensures that we process only jobs that have a reasonable high value (value $>$ planned energy investment) and that we do not awake from sleep for very cheap jobs.
The scheduling policy controls the speed, the job assignment, and the current mode of the processor.
It is a straightforward adaption of the algorithm used in~\cite{Han:2010}.
However, its analysis proves to be more involved because we have to take into account its interaction with the rejection policy and that the job sets scheduled by the optimal algorithm and $A$ may be quite different.

\begin{lstlisting}[language=pascal,keywords={arrival,depending,on,reject,working,idling,sleeping},caption={Rejection-oblivious online scheduler $A$.},label={lst:algorithmA},captionpos=b,mathescape=true]
{at any time $t$ and for $x$ equal to current idle cost}
on arrival of job $j$:
	{let $s_{\OA}$ be $\OA^t$'s speed for $j$ if it were accepted}
	reject if $\valdens{j}<\sfrac{\scrit^{\alpha-1}}{\alpha\crejTWO^{\alpha-1}}$ or $v_j<\crejONE x$ or $s_{\OA}>\crejTWO\sprofup{j}$

depending on current mode:
	{let $\rho_t$ denote $\OA^t$'s speed planned for for the current time $t$}
	working:
		if no remaining work:$\quad$ switch to idle mode
		otherwise:$\quad$ work at speed $\max(\rho_t,\scrit)$ with earliest deadline first
	idling:
		if $\disguisemath[l]{\rho_t>\scrit}{x\geq\gamma}$:$\quad$ switch to sleep mode
		if $\rho_t>\scrit$:$\quad$ switch to work mode
	sleeping:
		if $\rho_t>\scrit$:$\quad$ switch to work mode
\end{lstlisting}

The following description assumes a continuous recomputation of the current \OA schedule.
See Listing~\ref{lst:algorithmA} for the corresponding pseudocode.
It is straightforward to implement $A$ such that the planned schedule is recomputed only when new jobs arrive.

\paragraph{Scheduling Policy.}
All accepted jobs are scheduled according to the earliest deadline first rule.
At any time, the processor speed is computed based on the \OA schedule.
Use $\OA^t$ to denote the schedule produced by \OA if given the remaining (accepted) work at time $t$ and the power function $P_{\alpha,0}$.
Let $\rho_t$ denote the speed planned by $\OA^t$ at time $t$.
$A$ puts the processor either in \emph{working}, \emph{idling}, or \emph{sleeping} mode.
During working mode the processor speed is set to $\max(\rho_t,\scrit)$ until there is no more remaining work.
Then, speed is set to zero and the processor starts idling.
When idling or sleeping, we switch to the working mode only when $\rho_t$ becomes larger than \scrit.
When the amount of energy spent in the current idle interval equals the transition energy $\gamma$ (i.e., after time $\sfrac{\gamma}{P_{\alpha,\beta}(0)}$) the processor is suspended to sleep.

\paragraph{Rejection Policy.}
Let \crejONE and $\crejTWO$ be parameters to be determined later.
Consider the arrival of a new job $j$ at time $r_j$.
Reject it immediately if $\valdens{j}<\sfrac{\scrit^{\alpha-1}}{\alpha\crejTWO^{\alpha-1}}$.
Otherwise, define the \emph{current idle cost} $x\in\intcc{0,\gamma}$ depending on the processor's state as follows:
\begin{enumerate*}[(i)]
	\item zero if it is working,
	\item the length of the current idle interval times $\beta$ if it is idle, and
	\item $\gamma$ if it is asleep.
\end{enumerate*}
If $v_j<\crejONE x$, the job is rejected.
Otherwise, compute the job's speed $s_{\OA}$ which would be assigned by $\OA^{r_j}$ if it were accepted.
Reject the job if $s_{\OA}>\crejTWO\sprofup{j}$, accept otherwise.

\subsection{Bounding the Different Portions of the Cost}
In the following, let $O$ denote an optimal schedule.
Remember that $\cost{A}=\Esleep{A}+\Eidle{A}+\Ework{A}+\Vrej{A}$.
We bound each of the three terms $\Esleep{A}+\Eidle{A}$, $\Ework{A}$, and $\Vrej{A}$ separately in Lemma~\ref{lem:sleep+idle}, Lemma~\ref{lem:work}, and Lemma~\ref{lem:rejectedvalue}, respectively.
Eventually, Section~\ref{sec:competitiveness} combines these bounds and yields our main result: a nearly tight competitive factor depending on the maximum value density of the problem instance.

\begin{lemma}[Sleep and Idle Energy]\label{lem:sleep+idle}
$\Esleep{A}+\Eidle{A}\leq6\Esleep{O}+2\Esys{O}+\frac{4}{\crejONE}\Vrej{O}$
\end{lemma}
\begin{proof}
Let us first consider \Eidle{A}.
Partition the set of idle intervals under schedule $A$ into three disjoint subsets $\mathcal{I}_1$, $\mathcal{I}_2$, and $\mathcal{I}_3$ as follows:
\begin{itemize}
\item $\mathcal{I}_1$ contains idle intervals not intersecting any sleep interval of $O$.
	By definition, the total length of idle intervals from $\mathcal{I}_1$ is bounded by the time $O$ is awake.
	Thus, the total cost of $\mathcal{I}_1$ is at most \Esys{O}.
\item For each sleep interval $I$ of $O$, $\mathcal{I}_2$ contains any idle interval $X$ that is not the the last idle interval having a nonempty intersection with $I$ and that is completely contained within $I$ (note that the former requirement is redundant if the last intersecting idle interval is not completely contained in $I$).
	Consider any $X\in\mathcal{I}_2$ intersecting $I$ and let $j$ denote the first job processed by $A$ after $X$.
	It is easy to see that we must have $\intco{r_j,d_j}\subseteq I$.
	Thus, $O$ has rejected $j$.
	But since $A$ accepted $j$, we must have $v_j\geq\crejONE\abs{X}\beta$.
	This implies that the total cost of $\mathcal{I}_2$ cannot exceed $\sfrac{\Vrej{O}}{\crejONE}$.
\item $\mathcal{I}_3$ contains all remaining idle intervals.
	By definition, the first sleep interval of $O$ can intersect at most one such idle interval, while the remaining sleep intervals of $O$ can be intersected by at most two such idle intervals.
	Thus, if $m$ denotes the number of sleep intervals under schedule $O$, we get $\abs{\mathcal{I}_3}\leq2m-1$.
	Our sleeping strategy ensures that the cost of each single idle interval is at most $\gamma$.
	Using this and the definition of sleeping energy, the total cost of $\mathcal{I}_3$ is upper bounded by $(2m-1)\gamma=2\Esleep{O}+\gamma$.
\end{itemize}
Together, we get $\Eidle{A}\leq\Esys{O}+\sfrac{\Vrej{O}}{\crejONE}+2\Esleep{O}+\gamma$.
Moreover, without loss of generality we can bound $\gamma$ by $\sfrac{\Vrej{O}}{\crejONE}+\Esleep{O}$: if not both $A$ and $O$ reject all incoming jobs (in which case $A$ would be optimal), $O$ will either accept at least one job and thus wake up ($\gamma\leq\Esleep{O}$) or reject the first job $A$ accepts ($\gamma\leq\sfrac{\Vrej{O}}{\crejONE}$).
This yields $\Eidle{A}\leq\Esys{O}+2\sfrac{\Vrej{O}}{\crejONE}+3\Esleep{O}$.
For \Esleep{A}, note that any but the first of $A$'s sleep intervals is preceded by an idle interval of length $\sfrac{\gamma}{P_{\alpha,\beta}(0)}$.
Each such idle interval has cost $\gamma$, so we get $\Esleep{A}\leq\Eidle{A}$.
The lemma's statement follows by combining the bounds for $\Eidle{A}$ and $\Esleep{A}$.
\qed
\end{proof}

\begin{lemma}[Working Energy]\label{lem:work}
$\Ework{A}\leq\alpha^{\alpha}\Ework{O}+\crejTWO^{\alpha-1}\alpha^2\Vrej{O}$
\end{lemma}
The proof of Lemma~\ref{lem:work} is based on the standard amortized local competitiveness argument, first used by \textcite{Bansal:2007a}.
Although technically quite similar to the typical argumentation, our proof must carefully consider the more complicated rejection policy (compared to~\cite{Chan:2010}), while simultaneously handle the different processor states.

Given a schedule $S$, let $\Ework{S}(t)$ denote the working energy spent until time $t$ and $\Vrej{S}(t)$ the discarded value until time $t$.
We show that at any time $t\in\R_{\geq0}$ the \emph{amortized energy inequality}
\begin{equation}
\Ework{A}(t)+\Phi(t)\leq\alpha^{\alpha}\Ework{O}(t)+\crejTWO^{\alpha-1}\alpha^2\Vrej{O}(t)
\end{equation}
holds.
Here, $\Phi$ is a potential function to be defined in a suitable way.
It is constructed such that the following conditions hold:
\begin{enumerate}[(i)]
\item \emph{Boundary Condition:} At the beginning and end we have $\Phi(t)=0$.
\item \emph{Running Condition:} At any time $t$ when no job arrives we have
	\begin{equation}\label{eqn:differentialinequality}
	\od{\Ework{A}(t)}{t}+\od{\Phi(t)}{t}\leq\alpha^{\alpha}\od{\Ework{O}(t)}{t}+\crejTWO^{\alpha-1}\alpha^2\od{\Vrej{O}(t)}{t}
	.
	\end{equation}
\item \emph{Arrival Condition:} At any time $t$ when a job arrives we have
	\begin{equation}\label{eqn:differenceinequality}
	\Delta\Ework{A}(t)+\Delta\Phi(t)\leq\alpha^{\alpha}\Delta\Ework{O}(t)+\crejTWO^{\alpha-1}\alpha^2\Delta\Vrej{O}(t)
	.
	\end{equation}
	The $\Delta$-terms denote the corresponding change caused by the job arrival.
\end{enumerate}
Once these are proven, amortized energy inequality follows by induction:
It obviously holds for $t=0$, and Conditions (ii) and (iii) ensure that it is never violated.
Applying Condition (i) yields Lemma~\ref{lem:work}.
The crucial part is to define a suitable potential function $\Phi$.
Our analysis combines aspects from both \cite{Chan:2010} and \cite{Han:2010}.
Different rejection decisions of our algorithm $A$ and the optimal algorithm $O$ require us to handle possibly different job sets in the analysis, while the sleep management calls for a careful handling of the processor's current state.

\paragraph{Construction of $\Phi$.}
Consider an arbitrary time $t\in\R_{\geq0}$.
Let $\Wrem{A}{t}(t_1,t_2)$ denote the remaining work at time $t$ accepted by schedule $A$ with deadline in \intoc{t_1,t_2}.
We call the expression $\frac{\smash{\Wrem{A}{t}}(t_1,t_2)}{t_2-t_1}$ the \emph{density} of the interval $\intoc{t_1,t_2}$.
Next, we define \emph{critical intervals} \intoc{\tau_{i-1},\tau_i}.
For this purpose, set $\tau_0:=t$ and define $\tau_i$ iteratively to be the maximum time that maximizes the density $\rho_i:=\frac{\smash{\Wrem{A}{t}}(\tau_{i-1},\tau_i)}{\tau_i-\tau_{i-1}}$ of the interval $\intoc{\tau_{i-1},\tau_i}$.
We end at the first index $l$ with $\rho_l\leq\scrit$ and set $\tau_l=\infty$ and $\rho_l=\scrit$.
Note that $\rho_1>\rho_2>\ldots>\rho_l=\scrit$.
Now, for a schedule $S$ let $\smash{\Wrem{S}{t}}(i)$ denote the remaining work at time $t$ with deadline in the $i$-th critical interval \intoc{\tau_{i-1},\tau_i} accepted by schedule $S$.
The potential function is defined as $\Phi(t):=\alpha\sum_{i=1}^l\rho_i^{\alpha-1}\left(\Wrem{A}{t}(i)-\alpha\Wrem{O}{t}(i)\right)$.
It quantifies how far $A$ is ahead or behind in terms of energy.
The densities $\rho_i$ essentially correspond to \OA's speed levels, but are adjusted to $A$'s usage of \OA.
Note that whenever $A$ is in working mode its speed equals $\rho_1\geq\scrit$.

It remains to prove the boundary, running, and arrival conditions.
The boundary condition is trivially true as both $A$ and $O$ have no remaining work at the beginning and end.
For the running and arrival conditions, see Propositions~\ref{prop:runningcondition} and~\ref{prop:arrivalcondition}, respectively.
\begin{proposition}\label{prop:runningcondition}
The running condition holds.
That is, at any time $t$ when no job arrives we have
\begin{equation*}
\od{\Ework{A}(t)}{t}+\od{\Phi(t)}{t}\leq\alpha^{\alpha}\od{\Ework{O}(t)}{t}+\crejTWO^{\alpha-1}\alpha^2\od{\Vrej{O}(t)}{t}
.
\end{equation*}
\end{proposition}
\begin{proof}
Because no jobs arrive we have $\od{\Vrej{O}(t)}{t}=0$.
Let $s_A$ denote the speed of $A$ and $s_O$ the speed of $O$.
Depending on these speeds, we distinguish four cases:
\begin{description}
\item[Case 1:] $s_O=0$, $s_A=0$\\
	In this case $\od{\Ework{A}(t)}{t}=\od{\Ework{O}(t)}{t}=\od{\Phi(t)}{t}=0$.
	Thus, the Running Condition~\eqref{eqn:differentialinequality} holds.
\item[Case 2:] $s_O=0$, $s_A>0$\\
	Since $s_A>0$, algorithm $A$ is in working mode and we have $s_A=\rho_1\geq\scrit$.
	Moreover, $\od{\Ework{A}(t)}{t}=P_{\alpha,\beta}(s_A)$, $\smash{\od{\Phi(t)}{t}}=-\alpha s_A^{\alpha}$, and $\od{\Ework{O}(t)}{t}=0$.
	We get
	\begin{align*}
	& \od{\Ework{A}(t)}{t}+\od{\Phi(t)}{t}-\alpha^{\alpha}\od{\Ework{O}(t)}{t}=P_{\alpha,\beta}(s_A)-\alpha s_A^{\alpha}\\
	{}={}& \beta-(\alpha-1) s_A^{\alpha}	\leq\beta-(\alpha-1) \scrit^{\alpha}=0
	.
	\end{align*}
\item[Case 3:] $s_O>0$, $s_A=0$\\
	In this case $l=1$ and thus $\rho_1=\scrit$.
	The terms in Inequality~\eqref{eqn:differentialinequality} become $\od{\Ework{A}(t)}{t}=0$, $\od{\Phi(t)}{t}=\alpha^2\scrit^{\alpha-1} s_O$, and $\od{\Ework{O}(t)}{t}=P_{\alpha,\beta}(s_O)$.
	We get
	\begin{align*}
	&\textstyle \od{\Ework{A}(t)}{t}+\od{\Phi(t)}{t}-\alpha^{\alpha}\od{\Ework{O}(t)}{t}=\alpha^2\scrit^{\alpha-1} s_O-\alpha^{\alpha} P_{\alpha,\beta}(s_O)\\
	{}\leq{}&\textstyle \alpha^2\scrit^{\alpha-1} s_O-\alpha^{\alpha} s_O\frac{P_{\alpha,\beta}(\scrit)}{\scrit}\leq s_O\scrit^{\alpha-1}(\alpha^2-\alpha^{\alpha})\leq0
	.
	\end{align*}
\item[Case 4:] $s_O>0$, $s_A>0$\\
	Because of $s_A>0$ we know $A$ is in the working state and, thus, $s_A=\rho_1\geq\scrit$.
	So, this time we have $\od{\smash{\Ework{A}}(t)}{t}=P_{\alpha,\beta}(s_A)$, $\od{\Phi(t)}{t}=-\alpha s_A^{\alpha}+\alpha^2\rho_k^{\alpha-1}s_O$, and $\od{\smash{\Ework{O}}(t)}{t}=P_{\alpha,\beta}(s_O)$.
	We get
	\begin{align*}
	\od{\Ework{A}(t)}{t}+\od{\Phi(t)}{t}-\alpha^{\alpha}\od{\Ework{O}(t)}{t}&=P_{\alpha,\beta}(s_A)-\alpha s_A^{\alpha}+\alpha^2 \rho_k^{\alpha-1} s_O-\alpha^{\alpha} P_{\alpha,\beta}(s_O)\\
	&\leq s_A^{\alpha}-\alpha s_A^{\alpha}+\alpha^2 s_A^{\alpha-1} s_O-\alpha^{\alpha} s_O^{\alpha}\leq0
	.
	\end{align*}
	The last inequality follows from the same argument as in~\cite{Bansal:2007a}:
	Divide by $s_O^{\alpha}$ and substitute $z=\sfrac{s_A}{s_O}$.
	It becomes equivalent to $(1-\alpha)z^{\alpha}+\alpha^2z^{\alpha-1}-\alpha^{\alpha}\leq0$.
	Differentiating with respect to $z$ yields the correctness.
\end{description}
\qed
\end{proof}

\begin{proposition}\label{prop:arrivalcondition}
The arrival condition holds.
That is, at any time $t$ when a job arrives we have
\begin{equation*}
\Delta\Ework{A}(t)+\Delta\Phi(t)\leq\alpha^{\alpha}\Delta\Ework{O}(t)+\crejTWO^{\alpha-1}\alpha^2\Delta\Vrej{O}(t)
.
\end{equation*}
Here, the $\Delta$-terms denote the corresponding change caused by the job arrival.
\end{proposition}
\begin{proof}
The arrival of a job $j$ does not change the energy invested so far, thus $\Delta\Ework{A}(t)=0$ and $\Delta\Ework{O}(t)=0$.
If $A$ rejects $j$, we have $\Delta\Phi(t)\leq0$ and $\Delta\Vrej{O}(t)\geq0$, thus the Arrival Condition~\eqref{eqn:differenceinequality} holds.
So assume $A$ accepts $j$.
The arrival of $j$ may change the critical intervals and their densities significantly.
However, as pointed out in \cite{Bansal:2007a}, these changes can be broken down into a series of simpler changes affecting at most two adjacent critical intervals.
Thus, we first consider the effect of arrivals which do not change the critical intervals.
Afterward, we use the technique from~\cite{Bansal:2007a} to reduce an arbitrary change to these simple cases.
\begin{description}
\item[Case 1:] The critical intervals remain unchanged and only $\rho_k$ for $k<l$ changes.\\
	Let $\rho_k$ and $\rho_k'$ denote the densities of \intoc{\tau_{k-1},\tau_k} just before and after $j$'s arrival, respectively.
	That is, $\rho_k'=\rho_k+\smash{\frac{w_j}{\tau_k-\tau_{k-1}}}$.
	Note that $\rho_k'$ is the speed planned by \smash{$\OA^t$} for job $j$.
	Because $A$ accepted $j$ we have $\rho_k'\leq\crejTWO\sprofup{j}$.
	If $j$ is rejected by $O$, we have $\Delta\Vrej{O}(t)=v_j$.
	Since only the $k$-th critical interval is affected, the change in the potential function is given by
	\begin{equation*}
	\Delta\Phi(t)=\alpha\rho_k'^{\alpha-1}\left(\Wrem{A}{t}(k)+w_j-\alpha\Wrem{O}{t}(k)\right)-\alpha\rho_k^{\alpha-1}\left(\Wrem{A}{t}(k)-\alpha\Wrem{O}{t}(k)\right)
	.
	\end{equation*}
	Now, we compute, analogously to Lemma~4 in \cite{Chan:2010}, that $\Delta\Phi(t)$ equals
	\begin{align*}
	     &\textstyle \alpha\rho_k'^{\alpha-1}\left(\Wrem{A}{t}(k)+w_j-\alpha\Wrem{O}{t}(k)\right)-\alpha\rho_k^{\alpha-1}\left(\Wrem{A}{t}(k)-\alpha\Wrem{O}{t}(k)\right)\\
	{}\leq{} &\textstyle \alpha\rho_k'^{\alpha-1}\left(\Wrem{A}{t}(k)+w_j\right)-\alpha\rho_k^{\alpha-1}\Wrem{A}{t}(k)=\alpha\frac{\left(\Wrem{A}{t}(k)+w_j\right)^{\alpha}-\Wrem{A}{t}(k)^{\alpha}}{(\tau_k-\tau_{k-1})^{\alpha-1}}\\
	{}\leq{} &\textstyle \alpha^2\frac{\left(\Wrem{A}{t}(k)+w_j\right)^{\alpha-1}w_j}{(\tau_k-\tau_{k-1})^{\alpha-1}}=\alpha^2\rho_k'^{\alpha-1} w_j\leq\alpha^2(\crejTWO\sprofup{j})^{\alpha-1} w_j=\crejTWO^{\alpha-1}\alpha^2v_j
	.
	\end{align*}
	The penultimate inequality uses the fact that $f(x)=x^{\alpha}$ is convex, yielding $f(y)-f(x)\leq f'(y)\cdot(y-x)$ for all $y>x$.
	This implies the Arrival Condition~\eqref{eqn:differenceinequality}.
	If $j$ is accepted by $O$, we have $\Delta\Vrej{O}(t)=0$ and $\Delta\Phi(t)$ becomes
	\begin{equation*}
	\alpha\rho_k'^{\alpha-1}\left(\Wrem{A}{t}(k)+w_j-\alpha\left(\Wrem{O}{t}(k)+w_j\right)\right)-\alpha\rho_k^{\alpha-1}\left(\Wrem{A}{t}(k)-\alpha\Wrem{O}{t}(k)\right)
	.
	\end{equation*}
	In the same way as in~\cite{Bansal:2007a}, we now get $\Delta\Phi(t)\leq0$.
\item[Case 2:] Only the amount of work assigned to the last critical interval increases.\\
	Remember that, by definition, $\tau_l=\infty$ and $\rho_l=\scrit$.
	Since $j$ is accepted by $A$ we have $\scrit^{\alpha-1}\leq\alpha\crejTWO^{\alpha-1}\valdens{j}$.
	If $O$ rejects $j$ we have $\Delta\Vrej{O}(t)=v_j$ and the Arrival Condition~\eqref{eqn:differenceinequality} follows from
	\begin{align*}
	\Delta\Phi(t) &= \alpha\rho_l^{\alpha-1}\left(\Wrem{A}{t}(l)+w_j-\alpha\Wrem{O}{t}(l)\right)-\alpha\rho_l^{\alpha-1}\left(\Wrem{A}{t}(l)-\alpha\Wrem{O}{t}(l)\right)\\
	&= \alpha\scrit^{\alpha-1} w_j\leq\alpha^2\crejTWO^{\alpha-1}\valdens{j} w_j=\crejTWO^{\alpha-1}\alpha^2v_j
	.
	\end{align*}
	If $O$ accepts $j$, $\Delta\Vrej{O}(t)=0$ and the Arrival Condition~\eqref{eqn:differenceinequality} is implied by
	\begin{align*}
	\Delta\Phi(t) &= \alpha\rho_l^{\alpha-1}\left(\Wrem{A}{t}(l)+w_j-\alpha\left(\Wrem{O}{t}(l)+w_j\right)\right)-\alpha\rho_l^{\alpha-1}\left(\Wrem{A}{t}(l)-\alpha\Wrem{O}{t}(l)\right)\\
	&= \alpha\rho_l^{\alpha-1}w_j(1-\alpha)\leq0
	.
	\end{align*}
\end{description}
Let us now consider the arrival of an arbitrary job $j$.
The idea is to split this job into two jobs $j_1$ and $j_2$ with the same release time, deadline, and value density as $j$.
Their total work equals $w_j$.
Let $x$ denote the size of $j_1$.
We determine a suitable $x$ by continuously increasing $x$ from $0$ to $w_j$ until two critical intervals merge or one critical interval splits.
The arrival of $j_1$ can then be handled by one of the above cases, while $j_2$ is treated recursively in the same way as $j$.
For details, see~\cite{Bansal:2007a} or~\cite{Han:2010}.
\qed
\end{proof}

\subsubsection{Bounding the Rejected Value.}\label{sec:rejectedvalue}
In the following we bound the total value \Vrej{A} of jobs rejected by $A$.
The general idea is similar to the one by \textcite{Chan:2010}.
However, in contrast to the simpler model without sleep states, we must handle small-valued jobs of high density explicitly (cf.~Section~\ref{sec:lowerbounds}).
Moreover, the sleeping policy introduces an additional difficulty: our algorithm does not preserve all structural properties of an \OA schedule (cf.~Fact~\ref{fact:structuralproperties}).
This prohibits a direct mapping between the energy consumption of algorithm $A$ and of the intermediate \OA schedules during a \emph{fixed time interval}, as used in the corresponding proof in~\cite{Chan:2010}.
Indeed, the actual energy used by $A$ during a fixed time interval may decrease compared to the energy planned by the intermediate \OA schedule, as $A$ may decide to raise the speed to \scrit at certain points in the schedule.
Thus, to bound the value of a job rejected by $A$ but processed by the optimal algorithm for a relatively long time, we have to consider the energy usage for the workload \OA planned for that time (instead of the actual energy usage for the workload $A$ processed during that time, which might be quite different).
\begin{lemma}[Rejected Value]\label{lem:rejectedvalue}
Let \valdensMAX be the maximum value density of jobs of value less than $\crejONE\gamma$ and consider an arbitrary parameter $b\geq\sfrac{1}{\crejTWO}$.
Then, $A$'s rejected value is at most
\begin{equation*}
\Vrej{A}\leq\max\left(\valdensMAX\frac{\scrit}{P_{\alpha,\beta}(\scrit)},b^{\alpha-1}\right)\Ework{O}+\frac{b^{\alpha-1}}{(\crejTWO b-1)^{\alpha}}\Ework{A}+\Vrej{O}.
\end{equation*}
\end{lemma}
\begin{proof}
Partition the jobs rejected by $A$ into two disjoint subsets $J_1$ (jobs rejected by both $A$ and $O$) and $J_2$ (jobs rejected by $A$ only).
The total value of jobs in $J_1$ is at most $\Vrej{O}$.
Thus, it suffices to show that the total value of $J_2$ is bounded by
\begin{equation*}
\textstyle
\max\left(\valdensMAX\frac{\scrit}{P_{\alpha,\beta}(\scrit)},b^{\alpha-1}\right)\Ework{O}+\frac{b^{\alpha-1}}{(\crejTWO b-1)^{\alpha}}\Ework{A}
.
\end{equation*}
To this end, let $j\in J_2$.
Remember that, because of the convexity of the power function, $O$ can be assumed to process $j$ at a constant speed $s_O$.
Otherwise processing $j$ at its average speed could only improve the schedule.
Let us distinguish three cases, depending on the reason for which $A$ rejected $j$:
\begin{description}
\item[Case 1:] $j$ got rejected because of $\valdens{j}<\frac{\scrit^{\alpha-1}}{\alpha\crejTWO^{\alpha-1}}$.\\
	Let $\Ework{O}(j)$ denote the working energy invested by $O$ into job $j$.
	Using the rejection condition we can compute
	\begin{equation*}
	\textstyle
	\Ework{O}(j)=\frac{P_{\alpha,\beta}(s_O)}{s_O} w_j\geq\frac{P_{\alpha,\beta}(\scrit)}{\scrit} w_j\geq\scrit^{\alpha-1} w_j>\alpha\crejTWO^{\alpha-1} v_j
	.
	\end{equation*}
	Together with $b\geq\sfrac{1}{\crejTWO}$ we get $v_j<b^{\alpha-1}\Ework{O}(j)$.	
\item[Case 2:] $j$ got rejected because of $v_j<\crejONE x$\\
	As in the algorithm description, let $x\in[0,\gamma]$ denote the current idle cost at time $r_j$.
	Since $j$'s value is less than $\crejONE x\leq\crejONE\gamma$, we have $\valdens{j}\leq\valdensMAX$.
	We get
	\begin{equation*}
	\textstyle
	\Ework{O}(j)=\frac{P_{\alpha,\beta}(s_O)}{s_O} w_j=\frac{P_{\alpha,\beta}(s_O)}{s_O}\frac{v_j}{\valdens{j}}\geq\frac{P_{\alpha,\beta}(\scrit)}{\scrit}\frac{v_j}{\valdensMAX},
	\end{equation*}
	which eventually yields $v_j\leq\valdensMAX\frac{\scrit}{P_{\alpha,\beta}(\scrit)}\Ework{O}(j)$.
\item[Case 3:] $j$ got rejected because of $s_{\OA}>\crejTWO\sprofup{j}$\\
	Here, $s_{\OA}$ denotes the speed $\OA^{r_j}$ would assign to $j$ if it were accepted.
	We use $\OA^{r_j}_-$ to refer to the \OA schedule at time $r_j$ without $j$.
	Let $b_j:=\sfrac{\sprofup{j}}{s_O}$.
	We bound $v_j$ in different ways, depending on $b_j$.
	If $b_j$ is small (i.e., $b_j\leq b$) we use
	\begin{equation*}
	\textstyle
	\Ework{O}(j)\geq\frac{P_{\alpha,0}(s_O)}{s_O} w_j=\frac{P_{\alpha,0}(\sfrac{\sprofup{j}}{b_j})}{\sfrac{\sprofup{j}}{b_j}} w_j=\frac{\sprofup{j}^{\alpha-1}}{b_j^{\alpha-1}} w_j=\frac{v_j}{b_j^{\alpha-1}}.
	\end{equation*}
	That is, we have $v_j\leq b_j^{\alpha-1}\Ework{O}(j)$.
	Otherwise, if $b_j$ is relatively large, $v_j$ is bounded by $\Ework{A}$.
	Let $I$ denote the period of time when $O$ processes $j$ at constant speed $s_O$ and let $W$ denote the work processed by $\OA^{r_j}_-$ during this time.
	Since $I\subseteq\intco{r_j,d_j}$, Fact~\ref{fact:structuralproperties}(\ref{fact:structuralproperties_b}) implies that $\OA^{r_j}$'s speed during $I$ is at least $s_{\OA}>\crejTWO\sprofup{j}$.
	Thus, the total amount of work processed by $\OA^{r_j}$ during $I$ is larger than $\crejTWO\sprofup{j}\abs{I}$.
	But then, by applying Fact~\ref{fact:structuralproperties}(\ref{fact:structuralproperties_c}), we see that $W$ must be larger than $\crejTWO\sprofup{j}\abs{I}-w_j$.
	Now, $W$ is a subset of the work processed by $A$.
	Moreover, Fact~\ref{fact:structuralproperties}(\ref{fact:structuralproperties_d}) and the definition of algorithm $A$ ensure that the speeds used for this work in schedule $A$ cannot be smaller than the ones used in $\OA^{r_j}_-$.
	Especially, the average speed $s_{\varnothing}$ used for this work in schedule $A$ is at least $\sfrac{W}{\abs{I}}$ (the average speed used by $\OA^{r_j}_-$ for this work).
	Let $\Ework{A}(W)$ denote the energy invested by schedule $A$ into the work $W$.
	Then, by exploiting the convexity of the power function, we get
	\begin{align*}
	\Ework{A}(W) &\textstyle\geq \frac{P_{\alpha,\beta}(s_{\varnothing})}{s_{\varnothing}}W\geq\frac{P_{\alpha,0}(s_{\varnothing})}{s_{\varnothing}}W={s_{\varnothing}}^{\alpha-1}W\geq\frac{W^{\alpha-1}}{\abs{I}^{\alpha-1}}W=\abs{I}\frac{W^{\alpha}}{\abs{I}^{\alpha}}\\
	             &\textstyle>\abs{I}(\crejTWO\sprofup{j}-s_O)^{\alpha}=\frac{w_j}{s_O}s_O^{\alpha}(\crejTWO b_j-1)^{\alpha}=\frac{(\crejTWO b_j-1)^{\alpha}}{b_j^{\alpha-1}} v_j.
	\end{align*}
	That is, we have $v_j<\frac{b_j^{\alpha-1}}{(\crejTWO b_j-1)^{\alpha}}\Ework{A}(W)$.
	Now, let us specify how to choose from these two bounds:	
	\begin{itemize}
	\item If $b_j\leq b$, we apply the first bound: $v_j=b_j^{\alpha-1}\Ework{O}(j)\leq b^{\alpha-1}\Ework{O}(j)$.
	\item Otherwise we have $b_j>b\geq\sfrac{1}{\crejTWO}$.
		Note that for $x>\sfrac{1}{c}$ the function \smash{$f(x)=\frac{x^{\alpha-1}}{(cx-1)^{\alpha}}$} decreases.
		Thus, we get \smash{$v_j<\frac{b^{\alpha-1}}{(\crejTWO b-1)^{\alpha}}\Ework{A}(W)$}.
	\end{itemize}
\end{description}
By combining these cases we get
\begin{equation*}
\textstyle
v_j\leq\max\left(\valdensMAX\frac{\scrit}{P_{\alpha,\beta}(\scrit)},b^{\alpha-1}\right)\Ework{O}(j)+\frac{b^{\alpha-1}}{(\crejTWO b-1)^{\alpha}}\Ework{A}(W).
\end{equation*}
Note that both energies referred to, $\Ework{O}(j)$ as well as $\Ework{A}(W)$, are mutually different for different jobs $j$.
Thus, we can combine these inequalities for all jobs $j\in J_2$ to get the desired result.
\qed
\end{proof}

\subsection{Putting it All Together.}\label{sec:competitiveness}
The following theorem combines the results of Lemma~\ref{lem:sleep+idle}, Lemma~\ref{lem:work}, and Lemma~\ref{lem:rejectedvalue}.
\begin{theorem}\label{thm:competitiveness}
Let $\alpha\geq2$ and let \valdensMAX be the maximum value density of jobs of value less than $\crejONE\gamma$.
Moreover, define $\eta:=\max\bigl(\valdensMAX\frac{\scrit}{P_{\alpha,\beta}(\scrit)},b^{\alpha-1}\bigr)$ and $\mu:=\smash{\frac{b^{\alpha-1}}{(\crejTWO b-1)^{\alpha}}}$ for a parameter $b\geq\sfrac{1}{\crejTWO}$.
Then, $A$'s competitive factor is at most
\begin{equation*}
\max\left(\crejTWO^{\alpha-1}\alpha^2,\alpha^{\alpha}\right)\left(1+\mu\right)+\max\left(2+\eta,1+\sfrac{4}{\crejONE}\right)
.
\end{equation*}
\end{theorem}
\begin{proof}
Lemma~\ref{lem:sleep+idle} together with the relation $\Esys{O}\leq\Eidle{O}+\Ework{O}$ bounds the sleep and idle energy of $A$ with respect to $O$'s cost as $\Esleep{A}+\Eidle{A}\leq6\Esleep{O}+2\Eidle{O}+2\Ework{O}+\frac{4}{\crejONE}\Vrej{O}$.
For the working energy, Lemma~\ref{lem:work} yields $\Ework{A}\leq\alpha^{\alpha}\Ework{O}+\crejTWO^{\alpha-1}\alpha^2\Vrej{O}$.
To bound the total value rejected by $A$ with respect to the cost of $O$, we apply Lemma~\ref{lem:work} to Lemma~\ref{lem:rejectedvalue} and get
\begin{equation*}
\Vrej{A}\leq\eta\Ework{O}+\mu\Ework{A}+\Vrej{O}\leq \left(\eta+\alpha^{\alpha}\mu\right)\Ework{O}+\left(\crejTWO^{\alpha-1}\alpha^2\mu+1\right)\Vrej{O}
.
\end{equation*}
Using these inequalities, we can bound the cost of $A$ as follows:
\begin{align*}
\cost{A}\leq{} & 6\Esleep{O} + 2\Eidle{O} + \left(\alpha^{\alpha}+\alpha^{\alpha}\mu+2+\eta\right)\Ework{O}\\
               & + \left(\crejTWO^{\alpha-1}\alpha^2+\crejTWO^{\alpha-1}\alpha^2\mu+1+\sfrac{4}{\crejONE}\right)\Vrej{O}
.
\end{align*}
Since $6\leq\alpha^{\alpha}+2$ for $\alpha\geq2$, we get the following bound on $A$'s competitive factor:
\begin{equation*}
\textstyle
\frac{\cost{A}}{\cost{O}}\leq\max\left(\crejTWO^{\alpha-1}\alpha^2,\alpha^{\alpha}\right)\left(1+\mu\right)+\max\left(2+\eta,1+\sfrac{4}{\crejONE}\right)
.
\end{equation*}
\qed
\end{proof}
By a careful choice of parameters we get a constant competitive ratio if restricting the value density of small-valued jobs accordingly.
So, let $\alpha\geq2$ and set $\crejTWO=\alpha^{\frac{\alpha-2}{\alpha-1}}$, $b=\frac{\alpha+1}{\crejTWO}$, and $\crejONE=\frac{4}{1+b^{\alpha-1}}\leq1$.
Applying Theorem~\ref{thm:competitiveness} using these parameters yields the following results:
\begin{corollary}\label{cor:nicegeneralcompetitiveness}
Algorithm $A$ is $\alpha^{\alpha}+2 e\alpha+\valdensMAX\frac{\scrit}{P_{\alpha,\beta}(\scrit)}$-competitive.
\end{corollary}
\begin{corollary}\label{cor:competitiveness}
Algorithm $A$ is $\alpha^{\alpha}+2 e\alpha$-competitive if we restrict it to instances of maximum value density $\valdensMAX:=b^{\alpha-1}\frac{P_{\alpha,\beta}(\scrit)}{\scrit}$.
This competitive factor still holds if the restriction is only applied to jobs of value less than \smash{$\frac{4}{1+b^{\alpha-1}}\gamma$}.
\end{corollary}
\begin{proof}
First note the identity $b^{\alpha-1}=\alpha(1+\sfrac{1}{\alpha})^{\alpha-1}$.
Moreover, using the definitions from Theorem~\ref{thm:competitiveness}, we see that $\eta=b^{\alpha-1}$ and $\alpha^{\alpha}\mu=b^{\alpha-1}$.
By applying Theorem~\ref{thm:competitiveness} to our choice of parameters, the competitive factor of $A$ becomes
\begin{align*}
\alpha^{\alpha}(1+\mu)+2+\eta &= \alpha^{\alpha}+2+2 b^{\alpha-1}=\alpha^{\alpha} + 2\left(1+\alpha(1+\sfrac{1}{\alpha})^{\alpha-1}\right)\\
&\leq \alpha^{\alpha} + 2\alpha(1+\sfrac{1}{\alpha})^{\alpha}\leq\alpha^{\alpha}+2 e\alpha
.
\end{align*}
\qed
\end{proof}
\begin{corollary}\label{cor:competitivenessforcheapjobs}
If only considering instances for which the job values are at least $\frac{8}{2+3\alpha}\gamma\leq\gamma$, $A$'s competitive factor is at most $\alpha^{\alpha}+2e\alpha$.
\end{corollary}
\begin{proof}
Follows from Corollary~\ref{cor:competitiveness} by using that for $\alpha\geq2$ we have $\frac{4}{1+b^{\alpha-1}}=\frac{4}{1+\alpha(1+\sfrac{1}{\alpha})^{\alpha-1}}\leq\frac{4}{1+\frac{3}{2}\alpha}=\frac{8}{2+3\alpha}$.
\qed
\end{proof}
Note that the bound from Corollary~\ref{cor:nicegeneralcompetitiveness} is nearly tight with respect to $\valdensMAX$ and the lower bound from Theorem~\ref{thm:lowerbound}.

\section{The Speed-Bounded Case}
As stated earlier, our model can be considered as a generalization of~\cite{Chan:2010}.
It adds sleep states, leading to several structural difficulties which we solved in the previous section.
A further, natural generalization to the model is to cap the speed at some maximum speed $T$.
Algorithms based on \OA often lend themselves to such bounded speed models.
In many cases, a canonical adaptation – possibly mixed with a more involved job selection rule – leads to an algorithm for the speed bounded case with similar properties (see, e.g., \cite{Chan:2010,Han:2010,Chan:2007,Bansal:2008}).
A notable property of the profit-oriented scheduling model of~\cite{Chan:2010} is that limiting the maximum speed leads to a non-constant competitive factor.
Instead, it becomes highly dependent on a job's penalty ratio defined as $\penrat{j}:=\sfrac{\sprofup{j}}{T}$.
They derive a lower bound of $\LDAUOmega{\max(\sfrac{e^{\alpha-1}}{\alpha},\penratMAX^{\alpha-2+\sfrac{1}{\alpha}})}$ where $\penratMAX=\max\penrat{j}$.
Since our model generalizes their model, this bound transfers immediately to our setting (for the case $\beta=\gamma=0$).
On the positive side we can adapt our algorithm, similar to~\cite{Chan:2010}, by additionally rejecting a job if its speed planned by \OA is larger than $T$ (cf.~rejection condition in algorithm description, Section~\ref{sec:algorithm+analysis}).
Our main theorem from Section~\ref{sec:algorithm+analysis} becomes
\begin{theorem}\label{thm:competitiveness+boundedspeed}
Let $\alpha\geq2$ and let \valdensMAX be the maximum value density of jobs of value less than $\crejONE\gamma$.
Moreover, define $\eta:=\max\bigl(\valdensMAX\frac{\scrit}{P_{\alpha,\beta}(\scrit)},\penratMAX^{\alpha-1}b^{\alpha-1}\bigr)$ and $\mu:=\penratMAX^{\alpha-1}\smash{\frac{b^{\alpha-1}}{(b-1)^{\alpha}}}$ for $b\geq1$.
Then, $A$'s competitive factor is at most
\begin{equation*}
\alpha^{\alpha}\left(1+\mu\right)+\max\left(2+\eta,1+\sfrac{4}{\crejONE}\right)
.
\end{equation*}
\end{theorem}
\begin{proof}[sketch]
Note that the results from Lemmas~\ref{lem:sleep+idle} and~\ref{lem:work} remain valid without any changes, as an additional rejection rule does not influence the corresponding proofs.
The only lemma affected by the changed algorithm is Lemma~\ref{lem:rejectedvalue}.
In its proof, we have to consider an additional rejection case, namely that job $j$ got rejected because of $s_{\OA}>T=\frac{1}{\penrat{j}}\sprofup{j}$.
This can be handled completely analogously to Case~3 in the proof, using the factor $\frac{1}{\penrat{j}}$ instead of \crejTWO.
We get the bounds $v_j\leq b_j^{\alpha-1}\Ework{O}(j)$ and $v_j<\sfrac{b_j^{\alpha-1}}{(\sfrac{b_j}{\penrat{j}}-1)^{\alpha}}\Ework{A}(W)$.
If $b_j\leq\penrat{j}b$ this yields $v_j\leq \penrat{j}^{\alpha-1}b^{\alpha-1}\Ework{O}(j)$.
Otherwise, if $b_j>\penrat{j}b$, we have $v_j<\penrat{j}^{\alpha-1}\frac{b^{\alpha-1}}{(b-1)^{\alpha}}\Ework{A}(W)$.
The remaining argumentation is the same as in the proof of Theorem~\ref{thm:competitiveness}.
\qed
\end{proof}
For $b=\alpha+1$ and the interesting case $\Gamma>1$ we get a competitive factor of $\alpha^{\alpha}(1+2\Gamma^{\alpha-1})+\valdensMAX\frac{\scrit}{P_{\alpha,\beta}(\scrit)}$.
For job values of at most $\gamma$ it is $\alpha^{\alpha}(1+2\penratMAX^{\alpha-1})$.

\section{Conclusion \& Outlook}
We examined a new model that combines modern energy conservation techniques with profitability.
Our results show an inherent connection between the necessary and sufficient competitive ratio of rejection-oblivious algorithms and the maximum value density.
A natural question is how far this connection applies to other, more involved algorithm classes.
Can we find better strategies if allowed to reject jobs even after we invested some energy, or if taking former rejection decisions into account?
Such more involved rejection policies have proven useful in other models~\cite{Pruhs:2010,Han:2010}, and we conjecture that they would do so in our setting.
Other interesting directions include models for multiple processors as well as general power functions.
\Textcite{Pruhs:2010} modeled job values and deadlines in a more general way, which seems especially interesting for our profit-oriented model.

\printbibliography
\end{document}